\documentclass[nolimit]{article}
\usepackage{amsmath,amssymb,amsthm,mathrsfs,txfonts}
\usepackage{latexsym}
\usepackage{cite}
\usepackage[mathscr]{eucal}
\newtheorem{thm}{Theorem}[section]
\newtheorem{theorem}[thm]{Theorem}
\newtheorem{lemma}[thm]{Lemma}

\theoremstyle{definition}

\begin{document}
\date{}
\title{Tighter constraints of multiqubit entanglement for negativity}

\author{Long-Mei Yang$^1$, Bin Chen$^2$,
Shao-Ming Fei $^{3,4}$, Zhi-Xi Wang$^3$\thanks{Corresponding author: wangzhx@mail.cnu.edu.cn}
\\
{\footnotesize $^1$State Key Laboratory of Low-Dimensional Quantum Physics and Department of Physics, Tsinghua University, Beijing 100084, China}\\
{\footnotesize $^2$School of Mathematical Sciences, Tianjin Normal University, Tianjin 300387, China}\\
{\footnotesize $^3$School of Mathematical Sciences, Capital Normal University, Beijing 100048, China}\\
{\footnotesize $^4$Max-Planck-Institute for Mathematics in the Sciences, 04103 Leipzig, Germany}}

\maketitle

\begin{abstract}
We provide a characterization of multiqubit entanglement monogamy and polygamy constraints in terms of negativity.
Using the square of convex-roof extended negativity (SCREN) and the Hamming weight of the binary vector related to the distribution of subsystems
proposed in Kim (Phys Rev A 97: 012334, 2018), we provide a new class of monogamy inequalities of multiqubit entanglement based on the $\alpha$th power of SCREN for $\alpha\geq1$
and polygamy inequalities for $0\leq\alpha\leq1$ in terms of squared convex-roof extended negativity of assistance (SCRENoA).
For the case $\alpha<0$, we give the corresponding polygamy and monogamy relations for SCREN and SCRENoA, respectively.
We also show that these new inequalities give rise to tighter constraints than the existing ones.
\end{abstract}

\section{Introduction}

Quantum entanglement \cite{F.M,K.Chen,HPB1,HPB2,Vicente,CJZ} is one of the most intrinsic features of quantum mechanics,
which distinguishes the quantum from the classical world.
A distinct property of quantum entanglement is that a quantum system entangled with one of the other systems limits
its shareability with the remaining ones, known as the monogamy of entanglement (MoE) \cite{BMT,JSK}.
Being a useful resource, MoE plays a significant role in  many quantum information and communication processing tasks such as the security proof
in quantum cryptographic scheme \cite{CHB}.

For a given tripartite quantum state $\rho_{ABC}$, MoE can be characterized in a quantitative way known as monogamy inequality,
\begin{equation}
E(\rho_{ABC})\geq E(\rho_{AB})+E(\rho_{AC}),
\end{equation}
where $\rho_{AB}={\rm tr}_C(\rho_{ABC})$ and $\rho_{AC}={\rm tr}_B(\rho_{ABC})$ are the reduced density matrices.
In Ref. \cite{VC}, Coffman-Kundu-Wootters (CKW) established the first monogamy inequality based on the bipartite entanglement measure defined by tangle.
Later, Osborne {\it et al.} generalize the three-qubit CKW inequality to arbitrary multiqubit systems \cite{TJO}.
Monogamy inequalities in higher-dimensional quantum systems also have been deeply investigated by the use of various bipartite entanglement measures \cite{JSK1,JSK2,JSK3,JSK4}.

The assisted entanglement is a dual amount to bipartite entanglement measures, which accordingly has a dually monogamous property in multipartite quantum systems.
This dually monogamous property gives rise to a dual monogamy inequality known as polygamy inequality \cite{G.G1,G.G2}.
For a tripartite state $\rho_{ABC}$, one has
\begin{equation}\label{tauABC}
\tau^a(\rho_{A|BC})\leq\tau^a (\rho_{AB})+\tau^a(\rho_{AC}),
\end{equation}
where $\tau^a(\rho_{A|BC})$ is the tangle of assistance.

In Ref. \cite{JSK3, F.B}, the authors generalized the inequality \eqref{tauABC} to the cases of multiqubit quantum systems
and some class of higher-dimensional quantum systems.
By using the entanglement of assistance, a general polygamy inequality of multipartite entanglement in arbitrary-dimensional quantum systems has been also established \cite{JSK5,JSK6}.

Recently, based on the $\alpha$th power of entanglement measures, many generalized classes of monogamy inequalities were proposed \cite{Oliveira,Luo,SM.Fei1,SM.Fei2,SM.Fei3}.
In Ref. \cite{JSK7}, Kim investigated multiqubit entanglement constraints related to the negativity.
By using the $\alpha$th power of squared convex-roof extended negativity (SCREN) and the squared convex-roof extended negativity of assistance (SCRENoA)
for $\alpha\geq1$ and $0\leq\alpha\leq 1$, respectively, both monogamy and polygamy inequalities were established.
These inequalities involve the notion of Hamming weight of the binary vector related to the distribution of subsystems
and are shown to be tighter than the previous ones.

In this paper, we show that both the monogamy inequalities with $\alpha\geq1$ and polygamy inequalities with $0\leq\alpha\leq 1$
given in Ref. \cite{JSK7} can be further improved to be tighter.
Even for the case of $\alpha<0$, we can also provide tight constraints in terms of SCREN and SCRENoA.
Thus, a complete characterization for the full range of the power $\alpha$ is given.
These tighter constraints of multiqubit entanglement give rise to finer characterizations of the entanglement distributions
among the multiqubit systems.

\section{Preliminaries}

We first consider the monogamy inequalities and polygamy inequalities related to the negativity.
The tangle of a bipartite pure states $|\psi\rangle_{AB}$ is defined as \cite{VC}
\begin{equation}
\tau(|\psi\rangle_{A|B})=2(1-{\rm tr}\rho_A^2)
\end{equation}
where $\rho_A={\rm tr}_B|\psi\rangle_{AB}\langle\psi|$.
The tangle of a bipartite mixed state $\rho_{AB}$ is defined as
\begin{equation}\label{tauAB}
\tau(\rho_{A|B})=\Bigg[\min\limits_{\{p_k,|\psi_k\rangle\}}\sum\limits_{k}p_k\sqrt{\tau(|\psi_k\rangle_{A|B})}\Bigg]^2,
\end{equation}
and the tangle of assistance (ToA) of $\rho_{AB}$ is defined as
\begin{equation}\label{taua}
\tau^a(\rho_{A|B})=\Bigg[\max\limits_{\{p_k,|\psi_k\rangle\}}\sum\limits_{k}p_k\sqrt{\tau(|\psi_k\rangle_{A|B})}\Bigg]^2,
\end{equation}
where the minimization in \eqref{tauAB} and the maximum in \eqref{taua} are taken over all possible pure state decompositions of $\rho_{AB}=\sum\nolimits_{k}p_k|\psi_k\rangle_{AB}\langle\psi_k|$.

For any bipartite quantum state $\rho_{AB}$, the negativity is defined as \cite{JSK7,G.V},
 $\mathcal{N}(\rho_{A|B})=\|\rho_{AB}^{T_B}\|_1-1$,
where $\rho_{AB}^{T_B}$ is the partial transposition of $\rho_{AB}$, and $\|\cdot\|_1$ is the trace norm.
Then the notion of tangle and ToA for two-qubit state $\rho_{AB}$ in \eqref{tauAB} and \eqref{taua} can be rewritten as \cite{JSK7}
\begin{equation}
\tau(\rho_{A|B})=\Bigg[\min\limits_{\{p_k,|\psi_k\rangle\}}\sum\limits_{k}p_k\mathcal{N}(|\psi_k\rangle_{A|B})\Bigg]^2
\end{equation}
and
\begin{equation}
\tau^a(\rho_{A|B})=\Bigg[\max\limits_{\{p_k,|\psi_k\rangle\}}\sum\limits_{k}p_k\mathcal{N}(|\psi_k\rangle_{A|B})\Bigg]^2,
\end{equation}
respectively, due to the fact that $\mathcal{N}^2(|\psi\rangle_{A|B})=4\lambda_1\lambda_2=\tau(|\psi\rangle_{A|B})$
for any bipartite pure state $|\psi\rangle_{AB}$ with Schmidt rank 2,
$|\psi\rangle_{AB}=\sqrt{\lambda_1}|e_0\rangle_A\otimes|f_0\rangle_B+\sqrt{\lambda_2}|e_1\rangle_A\otimes|f_1\rangle_B$.

For higher-dimensional quantum systems, a rather natural generalization of two-qubit tangle is proposed, known as SCREN,
\begin{equation}\label{sc1}
\mathcal{N}_{sc}(\rho_{A|B})=\Bigg[\min\limits_{\{p_k,|\psi_k\rangle\}}\sum\limits_{k}p_k\mathcal{N}(|\psi_k\rangle_{A|B})\Bigg]^2.
\end{equation}
The dual quantity to SCREN can also be defined as
\begin{equation}\label{dualsc1}
\mathcal{N}_{sc}^a(\rho_{A|B})=\Bigg[\max\limits_{\{p_k,|\psi_k\rangle\}}\sum\limits_{k}p_k\mathcal{N}(|\psi_k\rangle_{A|B})\Bigg]^2,
\end{equation}
which is called the SCREN of assistance (SCRENoA).
Then, the tangle-based multiqubit monogamy and polygamy inequalities become as
\begin{equation}\label{ine1}
\mathcal{N}_{sc}(|\psi\rangle_{A_1|A_2\cdots A_n})\geq\sum\limits_{j=2}^n\mathcal{N}_{sc}(\rho_{A_1|A_j}),
\end{equation}
and
\begin{equation}\label{ine2}
\mathcal{N}_{sc}^a(|\psi\rangle_{A_1|A_2\ldots A_n})\leq\sum\limits_{j=2}^n\mathcal{N}_{sc}^a(\rho_{A_1|A_j}),
\end{equation}
where $\rho_{A_1|A_j}$ is two-qubit reduced density matrices $\rho_{A_1A_j}$ of subsystems $A_1A_j$ for $j=2,3,\ldots,n$ \cite{JSK7}.

Recently, Kim provided a class of monogamy and polygamy inequalities of multiqubit entanglement by the use of powered SCREN and the Hamming
weight of the binary vector related to the distribution of subsystems \cite{JSK7}.
For any nonnegative integer $j$ and its binary expansion $j=\sum\nolimits_{i=0}^{n-1}j_i2^i$,
where $\log_2^j<n$ and $j_i\in\{0,1\}$ for $i=0,1,\ldots,n-1$,
one can define a binary vector $\vec{j}$ as $\vec{j}=\{j_0,j_1,\ldots,j_{n-1}\}$.
The number of $1$'s in its coordinates is denoted as $\omega_H(\vec{j})$, called the Hamming weight of $\vec{j}$ \cite{MAN}.
Based on these notions, Kim proposed tight constraints of multiqubit entanglement as follows \cite{JSK7}:
\begin{equation}\label{Sc1}
[\mathcal{N}_{sc}(|\psi\rangle_{A|B_0B_1\ldots B_{N-1}})]^\alpha\geq \sum\limits_{j=0}^{N-1}\alpha^{\omega_{H}(\vec{j})}[\mathcal{N}_{sc}(\rho_{A|B_j})]^\alpha,
\end{equation}
for $\alpha\geq1$, and
\begin{equation}\label{ine3}
[\mathcal{N}_{sc}^a(|\psi\rangle_{A|B_0B_1\ldots B_{N-1}})]^\alpha\leq \sum\limits_{j=0}^{N-1}\alpha^{\omega_{H}(\vec{j})}[\mathcal{N}_{sc}^a(\rho_{A|B_j})]^\alpha,
\end{equation}
for $0\leq\alpha\leq1$.
Inequalities \eqref{Sc1} and \eqref{ine3} are then further written as:
\begin{equation}\label{Sc2}
[\mathcal{N}_{sc}(|\psi\rangle_{A|B_0B_1\ldots B_{N-1}})]^\alpha\geq \sum\limits_{j=0}^{N-1}\alpha^{j}[\mathcal{N}_{sc}(\rho_{A|B_j})]^\alpha,
\end{equation}
for $\alpha\geq1$, and
\begin{equation}\label{ine4}
[\mathcal{N}_{sc}^a(|\psi\rangle_{A|B_0B_1\ldots B_{N-1}})]^\alpha\leq \sum\limits_{j=0}^{N-1}\alpha^{j}[\mathcal{N}_{sc}^a(\rho_{A|B_j})]^\alpha,
\end{equation}
for $0\leq\alpha\leq1$.

However, these inequalities can be further improved to be much tighter under certain conditions, thus providing tighter constraints of multiqubit entanglement.

\section{Tighter constraints for SCREN}

In this section, we first provide a tighter monogamy inequality related to the $\alpha$th power of SCREN for $\alpha\geq1$.
For $\alpha<0$, a polygamy inequality is also proposed.
We need the following lemma.

\begin{lemma}\cite{Yang}
Suppose $k$ is a real number satisfying $0< k\leq1$, then for any $0\leq x\leq k$,
we have
\begin{equation}\label{Negativity1}
(1+x)^\alpha\geq1+\frac{(1+k)^\alpha-1}{k^\alpha}x^\alpha,
\end{equation}
for $\alpha\geq1$.
\end{lemma}

We have the following theorem.

\begin{theorem}\label{Nsc1}
For $\alpha\geq 1$ and any multiqubit pure state $|\psi\rangle_{AB_0\ldots B_{N-1}}$,
if the N-qubit subsystems $B_0, \ldots, B_{N-1}$ satisfy the following condition
\begin{equation}\label{order1}
k\mathcal{N}_{sc}(\rho_{A|B_j})\geq\mathcal{N}_{sc}(\rho_{A|B_{j+1}})\geq 0,
\end{equation}
where $j=0,1,\ldots,N-2$ and $0< k\leq1$,
then we have
\begin{equation}\label{nsc1}
[\mathcal{N}_{sc}(|\psi\rangle_{A|B_0B_1\ldots B_{N-1}})]^\alpha
\geq\sum\limits_{j=0}^{N-1}\Big(\frac{(1+k)^\alpha-1}{k^\alpha}\Big)^{\omega_H(\vec{j})}[\mathcal{N}_{sc}(\rho_{A|B_j})]^\alpha.
\end{equation}
\end{theorem}

\begin{proof}
Similar to the proof in \cite{JSK7}, from Eq. \eqref{ine1}, we only need to prove
\begin{equation}\label{nc1}
\left[\sum\limits_{j=0}^{N-1}\mathcal{N}_{sc}(\rho_{A|B_j})\right]^\alpha
\geq\sum\limits_{j=0}^{N-1}\Big(\frac{(1+k)^\alpha-1}{k^\alpha}\Big)^{\omega_{H}(\vec{j})}
[\mathcal{N}_{sc}(\rho_{A|B_j})]^\alpha.
\end{equation}
We first show that the inequality \eqref{nc1} holds for the case of $N=2^n$.
For $n=1$ and a three-qubit pure state $|\psi\rangle_{AB_0B_1}$,
from \eqref{Negativity1} and \eqref{order1}, one has
\begin{equation}
\begin{array}{rl}
&[\mathcal{N}_{sc}(\rho_{A|B_0})+\mathcal{N}_{sc}(\rho_{A|B_1})]^\alpha
=[\mathcal{N}_{sc}(\rho_{A|B_0})]^\alpha\Big(1+\frac{\mathcal{N}_{sc}(\rho_{A|B_1})}{\mathcal{N}_{sc}(\rho_{A|B_0})}\Big)^\alpha\\[4.0mm]
&\ \ \ \ \ \ \ \ \ \ \ \ \ \ \ \ \ \ \ \ \ \ \ \ \ \ \ \ \ \ \ \ \ \ \ \ \ \ \ \ \ \ \geq
[\mathcal{N}_{sc}(\rho_{A|B_0})]^\alpha \Bigg[1+\displaystyle\frac{(1+k)^\alpha-1}{k^\alpha}\Bigg(\frac{\mathcal{N}_{sc}(\rho_{A|B_1})}{\mathcal{N}_{sc}(\rho_{A|B_0})}\Bigg)^\alpha\Bigg]\\[4.0mm]
&\ \ \ \ \ \ \ \ \ \ \ \ \ \ \ \ \ \ \ \ \ \ \ \ \ \ \ \ \ \ \ \ \ \ \ \ \ \ \ \ \ \ =
[\mathcal{N}_{sc}(\rho_{A|B_0})]^\alpha+\displaystyle\frac{(1+k)^\alpha-1}{k^\alpha}[\mathcal{N}_{sc}(\rho_{A|B_1})]^\alpha,
\end{array}
\end{equation}
Thus, \eqref{nc1} holds for $n=1$.
Assume that inequality \eqref{nc1} holds for $N=2^{n-1}$ with $n\geq 1$. We consider the case of $N=2^n$.
For arbitrary $(N + 1)$-qubit pure state $|\psi\rangle_{AB_0B_1\ldots B_{N-1}}$ and its two-qubit reduced density matrices
$\rho_{AB_j}$, $j=0,1,\ldots,N-1$,
one has $\mathcal{N}_{sc}(\rho_{A|B_{j+2^{n-1}}})\leq k^{2^{n-1}}\mathcal{N}_{sc}(\rho_{A|B_j})$ from \eqref{order1}.
Then, we find
$$
0\leq\frac{\sum\nolimits_{j=2^{n-1}}^{2^n-1}\mathcal{N}_{sc}(\rho_{A|B_j})}{\sum\nolimits_{j=0}^{2^{n-1}-1}
\mathcal{N}_{sc}(\rho_{A|B_j})}\leq k^{2^{n-1}}\leq k,
$$
which implies that
\begin{equation}
\Bigg(1+\frac{\sum_{j=2^{n-1}}^{2^n-1}\mathcal{N}_{sc}(\rho_{A|B_j})}{\sum_{j=0}^{2^{n-1}-1}\mathcal{N}_{sc}
(\rho_{A|B_j})}\Bigg)^\alpha\geq 1+\displaystyle\frac{(1+k)^\alpha-1}{k^\alpha}
\Bigg(\frac{\sum_{j=2^{n-1}}^{2^n-1}\mathcal{N}_{sc}(\rho_{A|B_j})}
{\sum_{j=0}^{2^{n-1}-1}\mathcal{N}_{sc}(\rho_{A|B_j})}\Bigg)^\alpha.
\end{equation}
Thus,
\begin{equation}
\begin{array}{rl}
&\Bigg(\sum\nolimits_{j=0}^{N-1}\mathcal{N}(\rho_{A|B_j})\Bigg)^\alpha
=\Bigg(\sum\nolimits_{j=0}^{2^{n-1}-1}\mathcal{N}_{sc}(\rho_{A|B_j})+\sum\nolimits_{j=2^{n-1}}^{2^n-1}\mathcal{N}_{sc}(\rho_{A|B_j})\Bigg)^\alpha\\[4mm]
&\ \ \ \ \ \ \ \ \ \ \ \ \ \ \ \ \ \ \ \ \ \ \ \ \ \
=\Bigg(\sum\nolimits_{j=0}^{2^{n-1}-1}\mathcal{N}_{sc}(\rho_{A|B_j})\Bigg)^\alpha
\Bigg(1+\frac{\sum_{j=2^{n-1}}^{2^n-1}\mathcal{N}_{sc}(\rho_{A|B_j})}{\sum_{j=0}^{2^{n-1}-1}\mathcal{N}_{sc}
(\rho_{A|B_j})}\Bigg)^\alpha\\[4mm]
&\ \ \ \ \ \ \ \ \ \ \ \ \ \ \ \ \ \ \ \ \ \ \ \ \ \  \geq\Bigg(\sum\nolimits_{j=0}^{2^{n-1}-1}\mathcal{N}_{sc}(\rho_{A|B_j})\Bigg)^\alpha \Bigg[1+\displaystyle\frac{(1+k)^\alpha-1}{k^\alpha}
\Bigg(\frac{\sum_{j=2^{n-1}}^{2^n-1}\mathcal{N}_{sc}(\rho_{A|B_j})}
{\sum_{j=0}^{2^{n-1}-1}\mathcal{N}_{sc}(\rho_{A|B_j})}\Bigg)^\alpha\Bigg]\\[4mm]
&\ \ \ \ \ \ \ \ \ \ \ \ \ \ \ \ \ \ \ \ \ \ \ \ \ \  =\Bigg(\sum\nolimits_{j=0}^{2^{n-1}-1}\mathcal{N}_{sc}(\rho_{A|B_j})\Bigg)^\alpha
+\displaystyle\frac{(1+k)^\alpha-1}{k^\alpha}\Bigg(\sum\nolimits_{j=2^{n-1}}^{2^n-1}\mathcal{N}_{sc}(\rho_{A|B_j})\Bigg)^\alpha.
\end{array}
\end{equation}
Since we have assumed that
$$
\Bigg(\sum\nolimits_{j=0}^{2^{n-1}-1}\mathcal{N}_{sc}(\rho_{A|B_j})\Bigg)^\alpha\geq
\sum\nolimits_{j=0}^{2^{n-1}-1}\Big(\frac{(1+k)^\alpha-1}{k^\alpha}\Big)^{\omega_H(\vec{j})-1}[\mathcal{N}_{sc}(\rho_{A|B_j})]^\alpha,
$$
by relabeling the subsystems, we can always have
$$
\Bigg(\sum\nolimits_{j=2^{n-1}}^{2^n-1}\mathcal{N}_{sc}(\rho_{A|B_j})\Bigg)^\alpha\geq
\sum\nolimits_{j=2^{n-1}}^{2^n-1}\Big(\frac{(1+k)^\alpha-1}{k^\alpha}\Big)^{\omega_H(\vec{j})-1}[\mathcal{N}_{sc}(\rho_{A|B_j})]^\alpha.
$$
Then we have
$$
\Bigg(\sum\nolimits_{j=0}^{2^n-1}\mathcal{N}_{sc}(\rho_{A|B_j})\Bigg)^\alpha\geq\\[5mm]
\sum\nolimits_{j=0}^{2^n-1}\Big(\frac{(1+k)^\alpha-1}{k^\alpha}\Big)^{\omega_H(\vec{j})}[\mathcal{N}_{sc}(\rho_{A|B_j})]^\alpha.
$$

As there always exists an positive integer $n$ such that $0\leq N\leq 2^n$ for some positive integer $N$,
we consider a $(2^n+1)$-qubit pure state,
\begin{equation}\label{gamma1}
|\Gamma\rangle_{AB_0B_1\ldots B_{2^n-1}}=|\psi\rangle_{AB_0B_1\ldots B_{N-1}}\oplus |\phi\rangle_{B_N\ldots B_{2^n-1}},
\end{equation}
which is a product of $|\psi\rangle_{AB_0B_1\ldots B_{N-1}}$ and an arbitrary $(2^n-N)$-qubit pure state $|\phi\rangle_{B_N\ldots B_{2^n-1}}$ \cite{JSK7}.
Then we have
\begin{equation}
[\mathcal{N}_{sc}(|\Gamma\rangle_{AB_0B_1\ldots B_{2^n-1}})]^\alpha
\geq\sum\nolimits_{j=0}^{2^n-1}\Big(\frac{(1+k)^\alpha-1}{k^\alpha}\Big)^{\omega_H(\vec{j})}[\mathcal{N}_{sc}(\sigma_{A|B_j})]^\alpha
\end{equation}
with $\sigma_{A|B_j}$ being the two-qubit reduced density matrix of $|\Gamma\rangle_{AB_0B_1\ldots B_{2^n-1}}$ for each $j=0,1,\ldots,2^n-1$.
Thus,
\begin{equation}
\begin{array}{rl}
&[\mathcal{N}_{sc}(|\psi\rangle_{A|B_0B_1\ldots B_{N-1}})]^\alpha
=[\mathcal{N}_{sc}(|\Gamma\rangle_{A|B_0B_1\ldots B_{2^n-1}})]^\alpha\\[2.0mm]
&\ \ \ \ \ \ \ \ \ \ \ \ \ \ \ \ \ \ \ \ \ \ \ \ \ \ \ \ \ \ \ \ \ \ \ \geq\sum\nolimits_{j=0}^{2^n-1}\Big(\frac{(1+k)^\alpha-1}{k^\alpha}\Big)^{\omega_H(\vec{j})}[\mathcal{N}_{sc}(\sigma_{A|B_j})]^\alpha\\[2.0mm]
&\ \ \ \ \ \ \ \ \ \ \ \ \ \ \ \ \ \ \ \ \ \ \ \ \ \ \ \ \ \ \ \ \ \ \
=\sum\nolimits_{j=0}^{N-1}\Big(\frac{(1+k)^\alpha-1}{k^\alpha}\Big)^{\omega_H(\vec{j})}[\mathcal{N}_{sc}(\rho_{A|B_j})]^\alpha,
\end{array}
\end{equation}
since $|\Gamma\rangle_{A|B_0B_1\ldots B_{2^n-1}}$ is separable with respect to the bipartition between $AB_0\ldots B_{N-1}$ and $B_N\ldots B_{2^n-1}$.
\end{proof}

As $\Big(\frac{(1+k)^\alpha-1}{k^\alpha}\Big)^{\omega_H(\vec{j})}\geq\alpha^{\omega_H(\vec{j})}$ when $\alpha\geq1$,
we find that for any multiqubit pure state $|\psi\rangle_{A|B_0B_1\ldots B_{N-1}}$,
$[\mathcal{N}_{sc}(|\psi\rangle_{A|B_0B_1\ldots B_{N-1}})]^\alpha
\geq\sum\nolimits_{j=0}^{N-1}\Big(\frac{(1+k)^\alpha-1}{k^\alpha}\Big)^{\omega_H(\vec{j})}[\mathcal{N}_{sc}(\rho_{A|B_j})]^\alpha
\geq\sum\nolimits_{j=0}^{N-1}\alpha^{\omega_H(\vec{j})}[\mathcal{N}_{sc}(\rho_{A|B_j})]^\alpha$ with $\alpha\geq1$.
Thus,
inequality \eqref{nsc1} of Theorem \ref{Nsc1} is  tighter than inequality \eqref{Sc1} for any multiqubit pure state.

Here, we give an example to show that our new monogamy inequality is indeed tighter than the previous one given in \cite{JSK7}.

\bigskip
\noindent $\mathbf{Example} \ \ $
Let us consider a tripartite quantum state
\begin{equation}\label{psi}
|\psi\rangle_{ABC}=\frac{1}{\sqrt{6}}(|012\rangle-|021\rangle+|120\rangle-|102\rangle+|201\rangle-|210\rangle).
\end{equation}
Then we have $\mathcal{N}_{sc}(|\psi\rangle_{A|BC})=4$ and $\mathcal{N}_{sc}(|\psi\rangle_{A|B})=\mathcal{N}_{sc}(|\psi\rangle_{A|C})=1$ \cite{JSK7}.
Note that in this case $k=1$, and $[\mathcal{N}_{sc}(|\psi\rangle_{A|B})]^\alpha+\frac{(1+k)^\alpha-1}{k^\alpha}[\mathcal{N}_{sc}(|\psi\rangle_{A|c})]^\alpha
=1+\frac{(1+k)^\alpha-1}{k^\alpha}
=2^\alpha\geq
[\mathcal{N}_{sc}(|\psi\rangle_{A|B}]^\alpha+\alpha[\mathcal{N}_{sc}(|\psi\rangle_{A|c}]^\alpha=1+\alpha$ for $\alpha\geq1$.

Furthermore, by using Lemma 1, we can also improve inequality \eqref{nsc1} to be a tighter one under certain condition.

\begin{theorem}\label{Nsc2}
Suppose $k$ is a real number satisfying $0< k\leq1$.
For $\alpha\geq 1$ and any multiqubit pure state $|\psi\rangle_{AB_0\ldots B_{N-1}}$,
\begin{equation}\label{nsc2}
[\mathcal{N}_{sc}(|\psi\rangle_{A|B_0B_1\ldots B_{N-1}})]^\alpha
\geq\sum\nolimits_{j=0}^{N-1}\Big(\frac{(1+k)^\alpha-1}{k^\alpha}\Big)^j[\mathcal{N}_{sc}(\rho_{A|B_j})]^\alpha,
\end{equation}
if $k\mathcal{N}_{sc}(\rho_{A|B_j})\geq\sum\nolimits_{j=i+1}^{N-1}\mathcal{N}_{sc}(\rho_{A|B_j})$
for $i=0,1,\ldots, N-2$.
\end{theorem}
\begin{proof}
The proof is similar to the one given in \cite{JSK7}.
\end{proof}

In the next, we discuss the polygamy of entanglement related to the $\alpha$th power of SCREN for $\alpha<0$.
We have the following theorem.

\begin{theorem}\label{Nsc2}
For any multiqubit pure state $|\psi\rangle_{AB_0\ldots B_{N-1}}$ with $\mathcal{N}_{sc}(\rho_{AB_i})\neq0$,
$i=0,1,\ldots,N-1$,
we have
\begin{equation}\label{SCREN}
[\mathcal{N}_{sc}(|\psi\rangle_{A|B_0B_1\ldots B_{N-1}})]^\alpha
\leq\frac{1}{N}\sum\nolimits_{j=0}^{N-1}[\mathcal{N}_{sc}(\rho_{A|B_j})]^\alpha,
\end{equation}
for all $\alpha<0$.
\end{theorem}
\begin{proof}
We follow the proof given in \cite{SM.Fei2}.
For arbitrary tripartite state, we have
\begin{equation}\label{SCREN1}
[\mathcal{N}_{sc}(|\psi\rangle_{A|B_0B_1})]^\alpha
\leq[\mathcal{N}_{sc}(\rho_{A|B_0})+\mathcal{N}_{sc}(\rho_{A|B_1})]^\alpha
=\mathcal{N}_{sc}(\rho_{A|B_0})^\alpha\Big(1+\frac{\mathcal{N}_{sc}(\rho_{A|B_1})}{\mathcal{N}_{sc}(\rho_{A|B_0})}\Big)^\alpha
<[\mathcal{N}_{sc}(\rho_{A|B_0})]^\alpha,
\end{equation}
where the first inequality is due to $\alpha<0$ and the second inequality is due to
$\Big(1+\frac{\mathcal{N}_{sc}(\rho_{A|B_1})}{\mathcal{N}_{sc}(\rho_{A|B_0})}\Big)^\alpha<1$.
Similarly, we get
\begin{equation}\label{SCREN2}
[\mathcal{N}_{sc}(|\psi\rangle_{A|B_0B_1})]^\alpha<[\mathcal{N}_{sc}(\rho_{A|B_1})]^\alpha.
\end{equation}
From \eqref{SCREN1} and \eqref{SCREN2}, we obtain
\begin{equation}\label{SCREN3}
[\mathcal{N}_{sc}(|\psi\rangle_{A|B_0B_1})]^\alpha
<\frac{1}{2}\{[\mathcal{N}_{sc}(\rho_{A|B_0})]^\alpha+[\mathcal{N}_{sc}(\rho_{A|B_1})]^\alpha\}.
\end{equation}
One can get
\begin{equation}\label{SCREN4}
\begin{array}{rl}
&[\mathcal{N}_{sc}(|\psi\rangle_{A|B_0B_1\ldots B_{N-1}})]^\alpha
<\frac{1}{2}\{[\mathcal{N}_{sc}(\rho_{A|B_0})]^\alpha+[\mathcal{N}_{sc}(\rho_{A|B_1\ldots B_{N-1}})]^\alpha\}\\[2.0mm]
&\ \ \ \ \ \ \ \ \ \ \ \ \ \ \ \ \ \ \ \ \ \ \ \ \ \ \ \ \ \ \ \ \ \ \
<\frac{1}{2}[\mathcal{N}_{sc}(\rho_{A|B_0})]^\alpha+(\frac{1}{2})^2[\mathcal{N}_{sc}(\rho_{A|B_1})]^\alpha
+(\frac{1}{2})^2[\mathcal{N}_{sc}(\rho_{A|B_2\ldots B_{N-1}})]^\alpha\\[2.0mm]
&\ \ \ \ \ \ \ \ \ \ \ \ \ \ \ \ \ \ \ \ \ \ \ \ \ \ \ \ \ \ \ \ \ \ \ <\ldots \\[2.0mm]
&\ \ \ \ \ \ \ \ \ \ \ \ \ \ \ \ \ \ \ \ \ \ \ \ \ \ \ \ \ \ \ \ \ \ \
<\frac{1}{2}[\mathcal{N}_{sc}(\rho_{A|B_0})]^\alpha+(\frac{1}{2})^2[\mathcal{N}_{sc}(\rho_{A|B_1})]^\alpha+\ldots\\[2.0mm]
&\ \ \ \ \ \ \ \ \ \ \ \ \ \ \ \ \ \ \ \ \ \ \ \ \ \ \ \ \ \ \ \ \ \ \
+(\frac{1}{2})^{N-2}[\mathcal{N}_{sc}(\rho_{A|B_{N-2}})]^\alpha
+(\frac{1}{2})^{N-2}[\mathcal{N}_{sc}(\rho_{A|B_{N-1}})]^\alpha.
\end{array}
\end{equation}
By cyclically permuting the subindices $B_0$, $B_1,$ $\ldots$, $B_{N-1}$ in \eqref{SCREN4}, we can get a set of inequalities.
Summing up these inequalities, we have \eqref{SCREN}.
\end{proof}

\section{Tighter constraints for SCRENoA}

In this section, we provide a class of tighter polygamy inequalities of multiqubit entanglement
in terms of the $\alpha$-powered SCRENoA and the Hamming weight of the binary vector related to the distribution of subsystems
for $0\leq\alpha\leq1$.
For the case of $\alpha<0$, we also propose a monogamy relation for SCRENoA.

We need the following Lemma.

\begin{lemma}\cite{Yang}
Suppose $k$ is a real number satisfying $0< k\leq1$, then for any $0\leq x\leq k$,
we have
\begin{equation}\label{Negativity2}
(1+x)^\alpha\leq1+\frac{(1+k)^\alpha-1}{k^\alpha}x^\alpha,
\end{equation}
for $0\leq \alpha\leq 1$.
\end{lemma}

We have the following theorem.

\begin{theorem}\label{Nsc3}
Suppose $k$ is a real number satisfying $0< k\leq1$.
For $0\leq\alpha\leq1$ and any multiqubit pure state $|\psi\rangle_{AB_0\ldots B_{N-1}}$ satisfying
\begin{equation}\label{inequality3}
k\mathcal{N}_{sc}^a(\rho_{A|B_j})\geq\mathcal{N}_{sc}^a(\rho_{A|B_{j+1}})\geq0
\end{equation}
with $j=0,1,\ldots,N-2$, we have
\begin{equation}\label{nc6}
[\mathcal{N}_{sc}^a(|\psi\rangle_{A|B_0\ldots B_{N-1}})]^\alpha
\leq\sum\nolimits_{j=0}^{N-1}\Big(\frac{(1+k)^\alpha-1}{k^\alpha}\Big)^{\omega_H(\vec{j})}[\mathcal{N}_{sc}^a(\rho_{A|B_j})]^\alpha.
\end{equation}
\end{theorem}

\begin{proof}
From inequality \eqref{ine2}, we only need to show that
\begin{equation}\label{nc3}
\Bigg(\sum\nolimits_{j=0}^{N-1}\mathcal{N}_{sc}^a(\rho_{A|B_j})\Bigg)^\alpha
\leq\sum\nolimits_{j=0}^{N-1}\Big(\frac{(1+k)^\alpha-1}{k^\alpha}\Big)^{\omega_H(\vec{j})}[\mathcal{N}_{sc}^a(\rho_{A|B_j})]^\alpha.
\end{equation}
First, we prove inequality \eqref{nc3} for $N=2^n$.
For $n=1$ and a three-qubit pure state $|\psi\rangle_{AB_0B_1}$ with two-qubit reduced density $\rho_{AB_0}$ and $\rho_{AB_1}$, one has
\begin{equation}
\begin{array}{rl}
&[\mathcal{N}_{sc}^a(\rho_{A|B_0})+\mathcal{N}_{sc}^a(\rho_{A|B_1})]^\alpha
=[\mathcal{N}_{sc}^a(\rho_{A|B_0})]^\alpha\Big(1+\frac{\mathcal{N}_{sc}^a(\rho_{A|B_1})}{\mathcal{N}_{sc}^a(\rho_{A|B_0})}\Big)^\alpha\\[2.0mm]
&\ \ \ \ \ \ \ \ \ \ \ \ \ \ \ \ \ \ \ \ \ \ \ \ \ \ \ \ \ \ \ \ \ \ \ \ \ \ \ \ \ \
\leq[\mathcal{N}_{sc}^a(\rho_{A|B_0})]^\alpha \Bigg[1+\displaystyle\frac{(1+k)^\alpha-1}{k^\alpha}\Bigg(\frac{\mathcal{N}_{sc}^a(\rho_{A|B_1})}{\mathcal{N}_{sc}^a(\rho_{A|B_0})}\Bigg)^\alpha\Bigg]\\[2.0mm]
&\ \ \ \ \ \ \ \ \ \ \ \ \ \ \ \ \ \ \ \ \ \ \ \ \ \ \ \ \ \ \ \ \ \ \ \ \ \ \ \ \ \
=[\mathcal{N}_{sc}^a(\rho_{A|B_0})]^\alpha+\displaystyle\frac{(1+k)^\alpha-1}{k^\alpha}[\mathcal{N}_{sc}^a(\rho_{A|B_1})]^\alpha,
\end{array}
\end{equation}
where the inequality is due to \eqref{Negativity2}.
Assume \eqref{nc3} is true for $N=2^{n-1}$ with $n\geq1$.
We consider the case of $N=2^n$.
From \eqref{inequality3},
we find $\mathcal{N}_{sc}^a(\rho_{A|B_{j+2^{n-1}}})\leq k^{2^{n-1}}\mathcal{N}_{sc}^a(\rho_{A|B_j})$
for $j=0,1,\ldots,2^{n-1}-1$.
Then
$$0\leq\frac{\sum\nolimits_{j=2^{n-1}}^{2^n-1}\mathcal{N}_{sc}^a(\rho_{A|B_j})}{\sum\nolimits_{j=0}^{2^{n-1}-1}
\mathcal{N}_{sc}^a(\rho_{A|B_j})}\leq k^{2^{n-1}}\leq k.
$$
Thus,
\begin{equation}
\begin{array}{rl}
&\Bigg(\sum\nolimits_{j=0}^{N-1}\mathcal{N}_{sc}^a(\rho_{A|B_j})\Bigg)^\alpha
=\Bigg(\sum\nolimits_{j=0}^{2^{n-1}-1}\mathcal{N}_{sc}^a(\rho_{A|B_j})\Bigg)^\alpha
\Bigg(1+\frac{\sum_{j=2^{n-1}}^{2^n-1}\mathcal{N}_{sc}^a(\rho_{A|B_j})}{\sum_{j=0}^{2^{n-1}-1}\mathcal{N}_{sc}^a(\rho_{A|B_j})}\Bigg)^\alpha\\[4mm]
&\ \ \ \ \ \ \ \ \ \ \ \ \ \ \ \ \ \ \ \ \ \ \ \ \ \ \ \ \leq\Bigg(\sum\nolimits_{j=0}^{2^{n-1}-1}\mathcal{N}_{sc}^a(\rho_{A|B_j})\Bigg)^\alpha
\Bigg[1+\displaystyle\frac{(1+k)^\alpha-1}{k^\alpha}\Bigg(\frac{\sum_{j=2^{n-1}}^{2^n-1}\mathcal{N}_{sc}^a(\rho_{A|B_j})}
{\sum_{j=0}^{2^{n-1}-1}\mathcal{N}_{sc}^a(\rho_{A|B_j})}\Bigg)^\alpha\Bigg]\\[4mm]
&\ \ \ \ \ \ \ \ \ \ \ \ \ \ \ \ \ \ \ \ \ \  \ \ \ \ \ \ =\Bigg(\sum\nolimits_{j=0}^{2^{n-1}-1}\mathcal{N}_{sc}^a(\rho_{A|B_j})\Bigg)^\alpha
+\displaystyle\frac{(1+k)^\alpha-1}{k^\alpha}\Bigg(\sum\nolimits_{j=0}^{2^{n-1}-1}\mathcal{N}_{sc}^a(\rho_{A|B_j})\Bigg)^\alpha.
\end{array}
\end{equation}
Since we have assumed that
$$
\Bigg(\sum\nolimits_{j=0}^{2^{n-1}-1}\mathcal{N}_{sc}^a(\rho_{A|B_j})\Bigg)^\alpha\leq
\sum\nolimits_{j=0}^{2^{n-1}-1}\Big(\frac{(1+k)^\alpha-1}{k^\alpha}\Big)^{\omega_H(\vec{j})-1}[\mathcal{N}_{sc}^a(\rho_{A|B_j})]^\alpha,
$$
we obtain
$$
\Bigg(\sum\nolimits_{j=2^{n-1}}^{2^n-1}\mathcal{N}_{sc}^a(\rho_{A|B_j})\Bigg)^\alpha\leq
\sum\nolimits_{j=2^{n-1}}^{2^n-1}\Big(\frac{(1+k)^\alpha-1}{k^\alpha}\Big)^{\omega_H(\vec{j})-1}[\mathcal{N}_{sc}^a(\rho_{A|B_j})]^\alpha,
$$
Thus,
\begin{equation}
\begin{array}{rl}
&\Bigg(\sum\nolimits_{j=0}^{N-1}\mathcal{N}_{sc}^a(\rho_{A|B_j})\Bigg)^\alpha
\leq\sum\nolimits_{j=0}^{2^{n-1}-1}\Big(\frac{(1+k)^\alpha-1}{k^\alpha}\Big)^{\omega_H(\vec{j})}[\mathcal{N}_{sc}^a(\rho_{A|B_j})]^\alpha
+\frac{(1+k)^\alpha-1}{k^\alpha}
\sum\nolimits_{j=2^{n-1}}^{2^n-1}\Big(\frac{(1+k)^\alpha-1}{k^\alpha}\Big)^{\omega_H(\vec{j})-1}[\mathcal{N}_{sc}^a(\rho_{A|B_j})]^\alpha\\[4mm]
&\ \ \ \ \ \ \ \ \ \ \ \ \ \ \ \ \ \ \ \ \ \ \ \ \ \ \ \ =\sum\nolimits_{j=0}^{2^n-1}\Big(\frac{(1+k)^\alpha-1}{k^\alpha}\Big)^{\omega_H(\vec{j})}[\mathcal{N}_{sc}^a(\rho_{A|B_j})]^\alpha.
\end{array}
\end{equation}

For an arbitrary nonnegative integer $N$ and an $(N+1)$-qubit pure state $|\psi\rangle_{AB_0B_1\ldots B_{N-1}}$,
let us consider the $(2^n+1)$-qubit $|\Gamma\rangle_{AB_0B_1\ldots B_{N-1}}$ defined in \eqref{gamma1}.
We have
\begin{equation}
\begin{array}{rl}
&\mathcal{N}_{sc}^a(|\psi\rangle_{A|B_0B_1\ldots B_{N-1}})=\mathcal{N}_{sc}^a(|\Gamma\rangle_{A|B_0B_1\ldots B_{2^n-1}})\\[2.0mm]
&\ \ \ \ \ \ \ \ \ \ \ \ \ \ \ \ \ \ \ \ \ \ \ \ \ \ \ \ \ \ \ \
\leq\sum\nolimits_{j=0}^{2^n-1}\Big(\frac{(1+k)^\alpha-1}{k^\alpha}\Big)^{\omega_H(\vec{j})}[\mathcal{N}_{sc}^a(\sigma_{A|B_j})]^\alpha\\[2.0mm]
&\ \ \ \ \ \ \ \ \ \ \ \ \ \ \ \ \ \ \ \ \ \ \ \ \ \ \ \ \ \ \ \
=\sum\nolimits_{j=0}^{N}\Big(\frac{(1+k)^\alpha-1}{k^\alpha}\Big)^{\omega_H(\vec{j})}[\mathcal{N}_{sc}^a(\rho_{A|B_j})]^\alpha.
\end{array}
\end{equation}
\end{proof}

It can be seen that \eqref{nc6} is tighter than \eqref{ine3} since $\frac{(1+k)^\alpha-1}{k^\alpha}\leq\alpha$ for $0\leq\alpha\leq1$.

Moreover, the polygamy inequality of Theorem \ref{Nsc3} can be further improved under some conditions.

\begin{theorem}\label{Nsc4}
Suppose $k$ is a real number satisfying $0< k\leq1$.
For $0\leq\alpha\leq 1$ and any multiqubit pure state $|\psi\rangle_{AB_0\ldots B_{N-1}}$, we have
\begin{equation}\label{nc5}
[\mathcal{N}_{sc}^a(|\psi\rangle_{A|B_0\ldots B_{N-1}})]^\alpha
\leq\sum\nolimits_{j=0}^{N-1}\Big(\frac{(1+k)^\alpha-1}{k^\alpha}\Big)^j[\mathcal{N}_{sc}^a(\rho_{A|B_j})]^\alpha,
\end{equation}
if
\begin{equation}
k\mathcal{N}_{sc}^a(\rho_{A|B_i})\geq\sum\nolimits_{j=i+1}^{N-1}\mathcal{N}_{sc}^a(\rho_{A|B_j}),
\end{equation}
for $i=0,1,\ldots,N-2$.
\end{theorem}
\begin{proof}
The proof is similar to the one given in \cite{JSK7}.
\end{proof}

It should be noted that Theorems \ref{Nsc2} and \ref{Nsc4} provide the upper bound and the lower bound for
$\mathcal{N}_{sc}(|\psi\rangle_{A|B_o\ldots B_{N-1}})$,
since $\mathcal{N}_{sc}(|\psi\rangle_{A|B_o\ldots B_{N-1}})=\mathcal{N}_{sc}^a(|\psi\rangle_{A|B_o\ldots B_{N-1}})$.

The following lemma is useful for deriving monogamy relation in terms of $\alpha$-powered SCRENoA when $\alpha<0$.

\begin{lemma}
Suppose $k$ is a real number satisfying $0<k\leq 1$.
For $0\leq x\leq k$ and $\alpha<0$, we have
\begin{equation}\label{Negativity3}
(1+x)^\alpha\geq1+\frac{(1+k)^\alpha-1}{k^\alpha} x^\alpha.
\end{equation}
\end{lemma}
\begin{proof}
Let us consider the function $f(t,\alpha)=(1+t)^\alpha-t^\alpha$ with $t\geq\frac{1}{k}$ and $\alpha<0$.
Then $f_t(t,\alpha)=\alpha[(1+t)^{\alpha-1}-\alpha^{\alpha-1}]>0$,
i.e., $f(t,\alpha)$ is an increasing function with respect to $t$.
Thus,
\begin{equation}\label{Negativity4}
f(t,\alpha)\geq f\Big(\frac{1}{k},\alpha\Big)=\Big(1+\frac{1}{k}\Big)^\alpha-\frac{1}{k}=\frac{(1+k)^\alpha-1}{k^\alpha}.
\end{equation}
Set $x=\frac{1}{t}$ in \eqref{Negativity4}, we get \eqref{Negativity3}.
\end{proof}

\begin{theorem}\label{SCRENoA3}
Suppose $k$ is a real number satisfying $0<k\leq 1$.
For $\alpha<0$ and any multiqubit pure state $|\psi\rangle_{AB_0\ldots B_{N-1}}$, we have
\begin{equation}\label{SCRENoA2}
[\mathcal{N}_{sc}^a(|\psi\rangle_{A|B_0\ldots B_{N-1}})]^\alpha\geq\sum\nolimits_{j=0}^{N-1}\Big(\frac{(1+k)^\alpha-1}{k^\alpha}\Big)^j[\mathcal{N}_{sc}^a(\rho_{A|B_j})]^\alpha,
\end{equation}
if
\begin{equation}
k\mathcal{N}_{sc}^a(\rho_{A|B_i})\geq\mathcal{N}_{sc}^a(\rho_{A|B_{i+1}\ldots B_{N-1}})
\end{equation}
for $i=0,1,\ldots,N-2$.
\end{theorem}
\begin{proof}
From \eqref{ine2}, for arbitrary tripartite pure state $|\psi\rangle_{A|B_0B_1}$,
we get
\begin{equation}\label{SCRENoA1}
\begin{array}{rl}
&[\mathcal{N}_{sc}^a(|\psi\rangle_{A|B_0B_1})]^\alpha\geq[\mathcal{N}_{sc}^a(\rho_{A|B_0})+\mathcal{N}_{sc}^a(\rho_{A|B_1})]^\alpha\\[2.0mm]
&\ \ \ \ \ \ \ \ \ \ \ \ \ \ \ \ \ \ \ \ \ \ \ \ \ \ \ \
=[\mathcal{N}_{sc}^a(\rho_{A|B_0})]^\alpha\Big(1+\frac{\mathcal{N}_{sc}^a(\rho_{A|B_1})}{\mathcal{N}_{sc}^a(\rho_{A|B_0})}\Big)^\alpha\\[2.0mm]
&\ \ \ \ \ \ \ \ \ \ \ \ \ \ \ \ \ \ \ \ \ \ \ \ \ \ \ \
\geq[\mathcal{N}_{sc}^a(\rho_{A|B_0})]^\alpha+\frac{(1+k)^\alpha-1}{k^\alpha}[\mathcal{N}_{sc}^a(\rho_{A|B_1})]^\alpha.
\end{array}
\end{equation}
For arbitrary pure state $|\psi\rangle_{A|B_0\ldots B_{N-1}}$, we obtain
\begin{equation}
\begin{array}{rl}
&[\mathcal{N}_{sc}^a(|\psi\rangle_{A|B_0\ldots B_{N-1}})]^\alpha
\geq[\mathcal{N}_{sc}^a(\rho_{A|B_0})+\mathcal{N}_{sc}^a(\rho_{A|B_1\ldots B_{N-1}})]^\alpha\\[2.0mm]
&\ \ \ \ \ \ \ \ \ \ \ \ \ \ \ \ \ \ \ \ \ \ \ \ \ \ \ \ \ \ \ \
=[\mathcal{N}_{sc}^a(\rho_{A|B_0})]^\alpha\Big(1+\frac{\mathcal{N}_{sc}^a(\rho_{A|B_1\ldots B_{N-1}})}{\mathcal{N}_{sc}^a(\rho_{A|B_0})}\Big)^\alpha\\[2.0mm]
&\ \ \ \ \ \ \ \ \ \ \ \ \ \ \ \ \ \ \ \ \ \ \ \ \ \ \ \ \ \ \ \
\geq[\mathcal{N}_{sc}^a(\rho_{A|B_0})]^\alpha+\displaystyle\frac{(1+k)^\alpha-1}{k^\alpha}[\mathcal{N}_{sc}^a(\rho_{A|B_1\ldots B_{N-1}})]^\alpha\\[2.0mm]
&\ \ \ \ \ \ \ \ \ \ \ \ \ \ \ \ \ \ \ \ \ \ \ \ \ \ \ \ \ \ \ \ \geq\ldots\\[1.5mm]
&\ \ \ \ \ \ \ \ \ \ \ \ \ \ \ \ \ \ \ \ \ \ \ \ \ \ \ \ \ \ \ \
\geq[\mathcal{N}_{sc}^a(\rho_{A|B_0})]^\alpha+\displaystyle\frac{(1+k)^\alpha-1}{k^\alpha}[\mathcal{N}_{sc}^a(\rho_{A|B_1})]^\alpha+\ldots\\[2.0mm]
&\ \ \ \ \ \ \ \ \ \ \ \ \ \ \ \ \ \ \ \ \ \ \ \ \ \ \ \ \ \ \ \
+\Big(\frac{(1+k)^\alpha-1}{k^\alpha}\Big)^{N-1}[\mathcal{N}_{sc}^a(\rho_{A|B_{N-1}})]^\alpha,
\end{array}
\end{equation}
where the first inequality is due to $\alpha<0$, the second inequality is due to \eqref{Negativity3},
and the rest inequalities are due to \eqref{SCRENoA1}.
\end{proof}

Just like polygamy inequalities in Theorem 4 and Theorem 5, the following theorems give rise to the tighter monogamy relations in terms of
$\alpha$-powered SCRENoA for $\alpha<0$, with the notion of weighted constraint also involved.

\begin{theorem}\label{poly1}
Suppose $k$ is a real number satisfying $0< k\leq1$.
For $\alpha<0$ and any multiqubit pure state $|\psi\rangle_{AB_0\ldots B_{N-1}}$ satisfying
\begin{equation}\label{inequality4}
k\mathcal{N}_{sc}^a(\rho_{A|B_j})\geq\mathcal{N}_{sc}^a(\rho_{A|B_{j+1}})\geq0
\end{equation}
with $j=0,1,\ldots,N-2$, we have
\begin{equation}\label{nc7}
[\mathcal{N}_{sc}^a(|\psi\rangle_{A|B_0\ldots B_{N-1}})]^\alpha
\geq\sum\nolimits_{j=0}^{N-1}\Big(\frac{(1+k)^\alpha-1}{k^\alpha}\Big)^{\omega_H(\vec{j})}[\mathcal{N}_{sc}^a(\rho_{A|B_j})]^\alpha.
\end{equation}
\end{theorem}

\begin{theorem}\label{SCRENoA5}
Suppose $k$ is a real number satisfying $0< k\leq1$.
For $\alpha<0$ and any multiqubit pure state $|\psi\rangle_{AB_0\ldots B_{N-1}}$, we have
\begin{equation}
[\mathcal{N}_{sc}^a(|\psi\rangle_{A|B_0\ldots B_{N-1}})]^\alpha
\geq\sum\nolimits_{j=0}^{N-1}\Big(\frac{(1+k)^\alpha-1}{k^\alpha}\Big)^j[\mathcal{N}_{sc}^a(\rho_{A|B_j})]^\alpha,
\end{equation}
if
\begin{equation}
k\mathcal{N}_{sc}^a(\rho_{A|B_i})\geq\sum\nolimits_{j=i+1}^{N-1}\mathcal{N}_{sc}^a(\rho_{A|B_j}),
\end{equation}
for $i=0,1,\ldots,N-2$.
\end{theorem}

\section{Conclusion}
Entanglement monogamy is a fundamental property of multipartite entangled systems.
We have proposed tighter weighted monogamy inequalities related to the $\alpha$th power of SCREN for $\alpha\geq1$.
We also have investigated the polygamy relations in terms of $\alpha$-powered SCRENoA for the case of $0\leq\alpha\leq 1$.
Moreover, by using the $\alpha$th power of SCREN and SCRENoA for $\alpha<0$ respectively, the corresponding weighted polygamy and monogamy inequalities have also been established.
These new tighter monogamy and polygamy relations give rise to finer characterizations of the entanglement distributions
and capture better the intrinsic feature of multiparty quantum entanglement.

\section{Acknowledgements}
This work is supported by the National Natural Science Foundation of China under Nos. 11805143 and 11675113, and NSF of Beijing under No. KZ201810028042.


\begin{thebibliography}{00}
\bibitem{F.M} F. Mintert, Marek Ku\'s and A. Buchleitner, Concurrence of mixed bipartite quantum states in arbitrary dimensions,
Phys. Rev. Lett. $\mathbf{92}$, 167902 (2004).
\bibitem{K.Chen} K.Chen, S. Albeverio, and S.M. Fei, Concurrence of arbitrary dimensional bipartite quantum states,
Phys. Rev. Lett. \textbf{95}, 040504 (2005).
\bibitem{HPB1} H. P. Breuer, Separability criteria and bounds for entanglement measures, J. Phys. A Math. Gen. \textbf{39}, 11847 (2006).
\bibitem{HPB2}  H. P. Breuer, Optimal entanglement criterion for mixed quantum states, Phys. Rev. Lett. \textbf{97}, 080501 (2006).
\bibitem{Vicente} J. I. de Vicente, Lower bounds on concurrence and separability conditions, Phys. Rev. A \textbf{75}, 052320 (2007).
\bibitem{CJZ}  C. J. Zhang, Y. S. Zhang, S. Zhang, and G. C. Guo, Optimal entanglement witnesses based on local orthogonal observables,
Phys. Rev. A \textbf{76}, 012334 (2007).
\bibitem{BMT} B. M. Terhal, Is entanglement monogamous, IBM J. Res. Dev. \textbf{48}, 71 (2004).
\bibitem{JSK} J. S. Kim, G. Gour, and B. C. Sanders, Limitations to sharing entanglement, Contemp. Phys. \textbf{53}, 417 (2012).
\bibitem{CHB}  C. H. Bennett, Quantum cryptography using any two nonorthogonal states, Phys. Rev. Lett. \textbf{68}, 3121 (1992).
\bibitem{VC} V. Coffman, J. Kundu, and W. K. Wootters, Distributed entanglement, Phys. Rev. A \textbf{61}, 052306 (2000).
\bibitem{TJO} T. J. Osborne and F. Verstraete, General monogamy inequality for bipartite qubit entanglement,
Phys. Rev. Lett. \textbf{96}, 220503 (2006).
\bibitem{JSK1} J. S. Kim, A. Das, and B. C. Sanders,
Entanglement monogamy of multipartite higher-dimensional quantum systems using convex-roof extended negativity,
Phys. Rev. A \textbf{79}, 012329 (2009).
\bibitem{JSK2} J. S. Kimand B. C. Sanders, Monogamy of multi-qubit entanglement using R\'{e}nyi entropy,
 J. Phys.A Math. Theor. \textbf{43}, 445305 (2010).
\bibitem{JSK3} J. S. Kim, Tsallis entropy and entanglement constraints in multiqubit systems, Phys. Rev. A \textbf{81}, 062328 (2010).
\bibitem{JSK4} J. S. Kim and B. C. Sanders, Unified entropy, entanglement measures and monogamy of multi-party entanglement,
 J. Phys.A Math. Theor. \textbf{44}, 295303 (2011).
\bibitem{G.G1} G. Gour, D. A. Meyer, and B. C. Sanders, Deterministic entanglement of assistance and monogamy constraints,
Phys. Rev. A \textbf{72}, 042329 (2005).
\bibitem{G.G2} G. Gour, S. Bandyopadhay, and B. C. Sanders, Dual monogamy inequality for entanglement, J. Math. Phys.
\textbf{48}, 012108 (2007).
\bibitem{F.B} F. Buscemi, G. Gour, and J. S. Kim, Polygamy of distributed entanglement, Phys. Rev. A \textbf{80}, 012324
(2009).
\bibitem{JSK5} J. S. Kim, General polygamy inequality of multiparty quantum entanglement, Phys. Rev. A \textbf{85}, 062302 (2012).
\bibitem{JSK6} J. S. Kim, Tsallis entropy and general polygamy of multiparty quantum entanglement in arbitrary dimensions,
 Phys. Rev. A \textbf{94}, 062338 (2016).
\bibitem{Oliveira} T. R. de Oliveira, M. F. Cornelio, and F. F. Fanchini, Monogamy of entanglement of formation, Phys.
Rev. A \textbf{89}, 034303 (2014).
\bibitem{Luo} Y. Luo and Y. Li, Monogamy of $\alpha$th power entanglement measurement in qubit systems,
Ann. Phys. (NY) \textbf{362}, 511 (2015).
\bibitem{SM.Fei1} X. N. Zhu and S. M. Fei, Entanglement monogamy relations of qubit systems, Phys. Rev. A \textbf{90}, 024304 (2014).
\bibitem{SM.Fei2} Z. X. Jin and S. M. Fei, Tighter entanglement monogamy relations of qubit systems, Quantum Inf. Process. \textbf{16}, 77 (2017).
\bibitem{SM.Fei3} Z. X. Jin, J. Li, T. Li and S. M. Fei, Tighter monogamy relations in multiqubit systems,
Phys. Rev. A \textbf{97}, 032336 (2018).
\bibitem{JSK7} J. S. Kim, Negativity and tight constraints of multiqubit entanglement, Phys. Rev. A \textbf{97}, 012334 (2018).
\bibitem{G.V} G. Vidal and R. F. Werner, Computable measure of entanglement, Phys. Rev. A \textbf{65}, 032314 (2002).
\bibitem{MAN} M. A. Nielsen and I. L. Chuang, Quantum Computation and Quantum Information (Cambridge University Press, Cambridge, UK, 2000).
\bibitem{Yang} L.-M. Yang, B. Chen, S.-M. Fei, and Z.-X. Wang, Tighter constraints of multiqubit entanglement, Commun. Theor. Phys.  \textbf{71}, 545-554 (2019).

\end{thebibliography}
\end{document}